\documentclass[12pt]{article}

\usepackage{amsthm}
\usepackage[T1]{fontenc}
\usepackage{setspace} 

\usepackage{amsmath, amssymb,color}
\usepackage{accents}
\usepackage{mathtools}
\usepackage{amsfonts}
\usepackage{times}
\usepackage{tabularx}
\usepackage{enumerate}
\usepackage{comment}
\usepackage{subcaption}
\usepackage{bm}
\usepackage[english]{babel}
\usepackage{array}
\usepackage[margin=1in]{geometry}

\usepackage{authblk}
\usepackage{natbib}
\newcommand{\norm}[1]{\left\lVert#1\right\rVert}

\newtheorem{proposition}{Proposition}

\newtheorem{lemma}{Lemma}

\theoremstyle{definition}
\newtheorem{example}{Example}
\newtheorem{remark}{Remark}

\usepackage{graphicx,psfrag,epsf}
\usepackage{natbib}
\usepackage{url} 

\begin{document}



  \title{\bf A Note on Monte Carlo Integration in High Dimensions}
  \author{Yanbo Tang\thanks{
    The author gratefully acknowledge Heather Battey, Michaël Lalancette and Robert Zimmerman for helpful comments.}\hspace{.2cm}\\
    Department of Mathematics, Imperial College London}
  \maketitle

\bigskip
\begin{abstract}
  Monte Carlo integration is a commonly used technique to compute intractable integrals and is typically thought to perform poorly for very high-dimensional integrals. To show that this is not always the case, we examine Monte Carlo integration using techniques from the high-dimensional statistics literature by allowing the dimension of the integral to increase. In doing so, we derive non-asymptotic bounds for the relative and absolute error of the approximation for some general classes of functions through concentration inequalities.
  We provide concrete examples in which the magnitude of the number of points sampled needed to guarantee a consistent estimate vary between polynomial to exponential, and show that in theory arbitrarily fast or slow rates are possible. This demonstrates that the behaviour of Monte Carlo integration in high dimensions is not uniform.
  Through our methods we also obtain non-asymptotic confidence intervals for Monte Carlo estimate which are valid regardless of the number of points sampled.  
\end{abstract}

\section{Introduction}
Monte Carlo integration is a simple and widely used numerical integration method which leverages randomness and computational power to approximate intractable integrals. Chapters 1 and 2 of \cite{mcbook-owen} provide some historical details as well as an introduction to the subject. Roughly speaking, we can think of this method as generating a random grid of points on which we evaluate the function being integrated, and average these values to form an estimate of the integral.
As a result, Monte Carlo integration is generally thought to perform poorly for very high-dimensional functions.
The intuition behind this is clear: the curse of dimensionality makes covering the region of integration exponentially more difficult as the dimension of the integral increases.
However this poor scaling of the accuracy of the method is not immediate from the usual statement about the error rates of Monte Carlo methods.
While it is true that the approximation error is $O(n^{-1/2})$ \citep[Chapter 2.1]{mcbook-owen} on the absolute scale regardless of the dimension of the integrand, where $n$ is the sample size, the constant involved in this rate can dependent on the dimension of the integral. 
Furthermore, for many statistical applications, the relative accuracy is of interest rather than the absolute accuracy and these two metrics can behave quite differently.
In order to characterize the dependence of the constant on the dimension of the integral, we derive the order of $n$ needed to guarantee a non-trivial relative error as the dimension of the integral increases. 
Formally, we wish to integrate a sequence of functions $f_p(x)$ from $\mathbb{R}^p \rightarrow \mathbb{R}$ with respect to a sequence of probability densities $\phi_p(x)$. 
The Monte Carlo procedure uses the following approximation:
\begin{align}\label{eq:estimate}
 \int_{D_p} f_p(x)\phi_p(x) dx \approx \sum_{i = 1}^n \frac{f_p(X_i)}{n},
\end{align}
where $X_i$ for $i = 1. \dots, n$ are independent and identically distributed (iid) samples from the probability densities $\phi_p(x)$ with support $D_p \subset \mathbb{R}^p$; this is often taken to be the uniform measure if the domain has finite Lebesgue measure.

For the right side of (\ref{eq:estimate}) to be statistically useful, some form of uncertainty quantification is necessary. One commonly used measure of uncertainty is a $1 - \alpha$ level confidence interval. Keeping this in mind, along with the fact that we are interested in a measure of relative accuracy, we propose to
measure the difficulty of the problem by the relationship required between $p$ and $n(p)$ such that when $p \rightarrow \infty$, for every $\delta > 0$ and for every sequence $M_p \rightarrow \infty$:
\begin{align}\label{eq:consistency}
  \mathbb{P}\left( \left|\frac{\int_{D_p} f_p(x)\phi_p(x) dx - \sum_{i = 1}^{n(p)} f_p(X_i)/n(p) }{\int_{D_p} f_p(x)\phi_p(x) dx}\right| > M_p \delta \right) \rightarrow 0,
\end{align}
 where $n(p)$, the sample size, is an increasing function in $p$.
This is requiring $n(p)$ to grow as fast or faster than rate of increase necessary for 
 \[ 
  \mathbb{P}\left( \left|\frac{\int_{D_p} f_p(x)\phi_p(x) dx - \sum_{i = 1}^{n(p)} f_p(X_i)/n(p) }{\int_{D_p} f_p(x)\phi_p(x) dx}\right| > \delta \right) \rightarrow c \in (0, 1),
\]
 which guarantees a finite length confidence interval for the relative accuracy of the estimate as $p \rightarrow \infty$. This requirement is inspired from the definition of rates of posterior contraction \citep[Definition 8.1]{ghosal2017fundamentals} and the rates of consistency of point estimators. We assume that the expectations are non-zero in order to make condition (\ref{eq:consistency}) well-defined. The growth rate of the function $n(p)$ measures the complexity of the problem with respect to the dimensionality of the integral.  
We show in the sequel that this cost vary depending on the function being integrated, and can increase at an arbitrarily slow or fast speed, demonstrating the performance of Monte Carlo procedures is not uniformly affected by the dimension of the integral. 

Although considering a changing integrand may appear unusual at first, this type of scheme has been used as a theoretical device in the proof for the accuracy of quasi-Monte Carlo methods \citep{lemieux2009monte, owen2019monte}, where a carefully chosen deterministic grid is used instead of a randomly sampled grid.
In this setting, uncertainty quantification is obtained through the Koksma-Hlawka inequality \citep{koksma}, which is shown by constructing a adversarial function which changes with the number of grid points, $n$.
This regime is also common in high-dimensional statistics, in which a statistical model is allowed to increase in complexity as the sample size increases, see \citet{vershynin2018high,wainwright2019high} for an overview of high-dimensional probability and statistics respectively.
In our case, we can imagine that the function of interest lies somewhere along the sequence $f_p(x)$ and we wish to give some guidance on the required size of $n$ to satisfy a certain accuracy requirement, depending on the placement of this function.
Even though the statement in Equation (\ref{eq:consistency}) is asymptotic, we use non-asymptotic concentration bounds throughout this work. 
Through this approach we also obtain non-asymptotic $1- \alpha$ confidence intervals for the absolute error, and the relationship between the sampling cost of Monte Carlo procedures in high dimensions and the structure of $f_p(x)$ and $\phi_p(x)$. 

We first provide some background information on Monte Carlo integration and concentration inequalities.
We then examine some results for approximating volumes in high dimensions, and general integrals under structural assumptions.
We then provide a comparison between asymptotic and non-asymptotic 
confidence intervals in Section 5.
Calculations and derivations are deferred to the Appendix.
We take $g(n) = \omega(a_n)$ to mean that there exists a constant $N_0$ such that for all $n > N_0$ $|g(n)/a_n| \geq B$ for some $B > 0$, and take $ O(a_n)$ to mean that there exists a constant $N_0$ such that for all $n > N_0$ $|g(n)/a_n| \leq B$ for some $B > 0$
.
We use uppercase letters such as $X$ or $X_i$ to denote a vector-valued random variable. When invoking the Lipschitz property we take it to be with respect to the Euclidean norm. The Lipschitz property implies the uniform continuity of a function, and therefore bound its rate of increase uniformly across its domain. This fact will prove helpful in obtaining faster rates of concentration.

\section{Background}
Monte Carlo integration is justified by the application of the central limit theorem (CLT) in which we write
\begin{align*}
  \int_{D_p} f_p(x) \phi_p(x) dx = \mathbb{E}[f_p(X)], 
\end{align*}
where $\phi_p(\cdot)$ is a probability measure.
The addition of $\phi_p(\cdot)$ makes integration on infinite domains feasible; some texts take $\phi_p$ as the uniform probability measure when considering domains with finite Lebesgue measure.
If we are able to sample independent and identically distributed points $X_1, \dots, X_n$ from the density $\phi_p$, we can estimate the integral through:
\begin{align}\label{eq:monte_carlo_estimate}
  \mathbb{E}[f_p(X)] \approx  \frac{\sum_{i = 1}^n f_p(X_i)}{n},
\end{align}
and assuming that $E[|f_p(X)|] < \infty$ then we are guaranteed a consistent estimate as $n \rightarrow \infty$ by the strong law of large numbers. 

The approximation error is typically obtained by using an empirical estimate of the variance and appealing to the CLT to produce an asymptotically valid confidence interval. 
However, noting that the estimate introduced in (\ref{eq:monte_carlo_estimate}) is the average, quantifying the approximation error is equivalent to quantifying how rapidly an average concentrates around its expectation.
Seeing the problem from this perspective allows us to use powerful concentration inequalities to bound the non-asymptotic approximation error, see \cite{boucheron2013concentration} for a textbook discussing concentration inequalities and their uses.
In order to use these bounds however, we require additional assumptions on $f_p(x)$, and $\phi_p(x)$. We introduce some particular structures which lead to rapid concentration, including averages of independent sub-Gaussian random variables, concave/convex Lipschitz functions on the hypercube and Lipschitz functions of strongly log-concave random variables.

\subsection{Averages of Subgaussian Variables}
Sub-Gaussian random variables are a class of random variables whose tail probabilities can be dominated by that of a Gaussian random variable with mean $\mu$ and variance $\sigma^2$.
As Gaussian distributions concentrated around their mean at an exponentially fast rate \citep[Theorem 1.2.6]{durrett2019probability}, we can consequently expect the same behavior from sub-Gaussian random variables.     
It can be shown that for a random variable to have lighter tail than a Gaussian distribution with mean $\mu$ and variance $\sigma^2 >0$ is mathematically equivalent to upper bounding the moment generating function of the centered random variable by that of a Gaussian with variance $\sigma^2$ \citep[Theorem 2.6]{wainwright2019high}.
More precisely, a random variable is said to be sub-Gaussian with mean $u$ and a proxy variance $\sigma^2 > 0$ if its moment generating function $M_Z(t)$ satisfies: 
\[
 M_Z(t) = E\{\exp[t(Z - \mu)]\} \leq \exp\left( \frac{\sigma^2 t^2}{2} \right) \text{ for all } t \in \mathbb{R},
\]
note that the proxy variance $\sigma^2$ may not be unique in this definition, but taking $\sigma^2$ to be as small as possible while satisfying the above condition results in a sharper bound.
Then $\bar{Z}$, the average of iid realizations $Z_1, \dots, Z_n$ satisfies
\begin{align}\label{eq:subG_Tail}
  \mathbb{P}\left[ \left| \bar{Z} - E(Z) \right| > \delta  \right] \leq 2\exp\left(-\frac{\delta^2 n}{2\sigma^2} \right),  
\end{align}
see \citet[\S2]{wainwright2019high}. For the purpose of verifying condition (\refeq{eq:consistency}), the bound in Equation (\ref{eq:subG_Tail}) can be re-expressed as:
\begin{align}\label{eq:subG_Tail_rel}
  \mathbb{P}\left[ \left| \frac{\bar{Z} - E(Z)}{E(Z)} \right| > \delta  \right] \leq 2\exp\left[-\frac{\delta^2 E(Z)^2 n}{2\sigma^2} \right],  
\end{align}
through simple algebraic manipulations. This form is more useful for studying the relative accuracy of Monte Carlo integration. We assume the expectation is non-zero or the ratio in (\ref{eq:subG_Tail_rel}) is ill-defined.
As we do not know the expectation, we need to bound it from below in order to use (\ref{eq:subG_Tail_rel}) in practice; for convex functions this can sometimes be accomplished by Jensen's inequality.
One important class of subgaussian random variables are bounded random variables taking value in the interval $[a, b]$, in which case the proxy variance $\sigma^2$ is $(b - a)^2/4$, this is a consequence of Hoeffding's Lemma \citep[Lemma 2.2]{boucheron2013concentration}.

\subsection{Convex/Concave Lipschitz Functions of Random Variables on the Unit Cube}
Beyond simple averages, there are concentration results for convex/concave Lipschitz functions of random variables on the hypercube. 
For a Lipschitz convex/concave function $g(x)$ with Lipchitz constant $L$, and for a random variable $X \subset [0,1]^p$ whose components are independently distributed, 
\begin{align}
  \mathbb{P}\left\{ \left| g(X) - E[g(X)] \right| > \delta  \right\} \leq 2\exp\left(-\frac{\delta^2}{2L^2} \right).  
\end{align}
which implies
\begin{align}
  \mathbb{P}\left\{ \left| \frac{g(X) - E[g(X)]}{E[g(X)]} \right| > \delta  \right\} \leq 2\exp\left\{-\frac{\delta^2 E[g(X)]^2 }{2L^2} \right\},  
\end{align}
by Theorem 3.24 in \cite{wainwright2019high}.
We will apply this bound to averages of convex/concave Lipschitz functions in Section 4, by using the fact that averages of such functions remains convex/concave Lipschitz. The surprising aspect of this bound is that it is \textit{dimension free}, as it holds for all convex/concave Lipschitz continuous functions of arbitrary dimension. However, as we will see in the examples in Section 4, it is possible that the Lipschitz constant $L$ increases with the dimension of the function.

\subsection{Lipschitz Functions of Strongly Log-concave Random Variables}
For some $\gamma> 0$, a density $\phi(x) = \exp[-\psi(x)]$ is strongly $\gamma$-log-concave if it satisfies the following: 
\begin{align*}
  \lambda\psi(x) + (1- \lambda)\psi(y) - \psi[ \lambda x + (1 - \lambda)y  ] \geq \frac{\gamma}{2} \lambda(1- \lambda)\norm{x - y}_2^2, 
\end{align*}
for all $\lambda \in [0,1]$ and all $x, y \in \mathbb{R}^p$.
Then for any $L$-Lipschitz function $g$ and a strongly $\gamma$-log-concave random variable $X$:
\begin{align*}
  \mathbb{P}\left\{ \left| g(X) - E[g(X)] \right| > \delta  \right\} \leq 2\exp\left(-\frac{\gamma\delta^2}{4L^2} \right),  
\end{align*}
and
\begin{align*}
  \mathbb{P}\left\{ \left| \frac{g(X) - E[g(X)]}{E[g(X)]} \right| > \delta  \right\} \leq 2\exp\left\{-\frac{\gamma\delta^2 E[g(X)]^2 }{4L^2} \right\},
\end{align*}
by Theorem 3.16 in \cite{wainwright2019high}.
Strongly log-concave distributions can also be characterized through the product of a log-concave density with a multivariate normal distribution \citep{saumard2014log}.
We will apply this bound to averages of iid realizations of g(X) by using properties of strongly $\gamma$-log-concave random variables in Section 4. This bound also does not depend on the dimension of the integral, although once again, the constants involved can potentially increase with $p$. Intuitively, the strongly log-concave density must be uni-modal and `hump-like' on the log scale, and as this is exponentiated the resulting density concentrates rapidly around its mean.

\section{Estimation of High-Dimensional Volumes}
Estimation of volumes is a classic application of Monte Carlo integration. 
Suppose we are interested in calculating the unknown volume (Lebesgue measure) of a set $F$ which is contained within a set $F^\prime$ with known finite volume.
Then if we can sample $X_i$'s from the uniform distribution supported on $F^\prime$, we can approximate the volume of $F$ by:
\begin{align*}
  \text{Vol}(F) \approx \text{Vol}(F^\prime) \sum_{i = 1}^n \frac{\mathbb{I}( X_i \in F )}{n}.
\end{align*}
As the random variable $\text{Vol}(F^\prime)\mathbb{I}( X_i \in F )$ only takes the value of $0$ or $\text{Vol}(F^\prime)$, it is subgaussian with $\sigma^2 = \text{Vol}(F^\prime)^2/4$. Noting this and using the subgaussian bound in Section 2.1 gives:  

\begin{proposition}\label{prop:volume}
  Suppose that both $\text{Vol}(F)$ and $\text{Vol}(F^\prime)$ are finite. Then:
  \begin{align*}
    \mathbb{P}\left[ \left|  \frac{{\text{Vol}}(F^\prime)\sum_{i = 1}^n \mathbb{I}( X_i \in F )}{n} - \text{Vol}(F) \right| > \delta  \right] \leq 2\exp\left[-\frac{2\delta^2 n}{\text{Vol}(F^\prime )^2} \right],    
  \end{align*}
  which implies the non-asymptotic $1- \alpha$ confidence interval 
\[ 
  \left(  \text{Vol}(F^\prime)\left[ \frac{\sum_{i = 1}^n \mathbb{I}( X_i \in F )}{n} - \sqrt{\frac{\log(2/\alpha)}{2 n}} \right],\ \text{Vol}(F^\prime)\left[ \frac{\sum_{i = 1}^n \mathbb{I}( X_i \in F )}{n} + \sqrt{\frac{\log(2/\alpha)}{2 n}}\right]  \right) 
\]
for $\alpha < 1$ and
\begin{align*}
  \mathbb{P}\left[ \left| \frac{ \text{Vol}(F^\prime) \sum_{i = 1}^n \mathbb{I}( X_i \in F )/n - \text{Vol}(F)}{\text{Vol}(F)} \right| > \delta  \right] \leq 2\exp\left[-\frac{2\delta^2 \text{Vol}(F)^2 n}{\text{Vol}(F^\prime )^2} \right].
\end{align*}
where $X_i$ are iid copies of a uniform random variable supported on $F^\prime$.
\end{proposition}

The typical example provided in textbooks such as \citet[Chapter 3.2]{ross2013simulation}, is the problem of computing the value of $\pi$ or the area of a circle inscribed in a square. 
We extend the latter to the case of a hypersphere contained in a hypercube and showcase how the curse of dimensionality manifests itself in volume estimation.

\begin{example}
  Let $F = \{ x \in R^p : \norm{x}_2 \leq 1 \}$ and let $F^\prime = \{ x \in R^p : \norm{x}_\infty \leq 1 \}$, where $F$ is a hypersphere of dimension $p$ with radius 1, 
  while $F^\prime$ is the hypercube centered at $0$ with sides of length 2. 
  In this case we know that $\text{Vol}(F) = \pi^{p/2}/\Gamma(p/2 +1 )$ and $\text{Vol}(F^\prime) = 2^p$, thus
  \begin{align*}
    &\mathbb{P}\left[ \left| {\text{Vol}}(F^\prime) \sum_{i = 1}^n \mathbb{I}( X_i \in F )/n - \text{Vol}(F) \right| > \delta  \right] \leq 2\exp\left(-\frac{\delta^2 n}{2^{2p - 1}} \right), \\
    &\mathbb{P}\left[ \left| \frac{ \text{Vol}(F^\prime) \sum_{i = 1}^n \mathbb{I}( X_i \in F ) - \text{Vol}(F)}{\text{Vol}(F)} \right| > \delta  \right] \leq 2\exp\left[-\frac{\delta^2 \pi^p n}{2^{2p -1} \Gamma(p/2 + 1)^2} \right].
  \end{align*}
  The absolute error is decaying extremely quickly, as the volume of the unit sphere decays exponentially to $0$ as the dimension increases, therefore the chance of hitting the sphere by sampling from the unit cube is also exponentially decaying to 0. Our estimate will essentially be $0$, but this is quite close to the volume of an unit sphere in high dimensions.
  However, in order to satisfy our relative consistency requirement, we need the number of $X_i$ sampled to be
  $n(p) = \exp\{\omega[p\log(p)]\}$,  
  by Stirling's approximation. Therefore the numerical cost of integrating this function is exponentially increasing in the dimension of the sphere.
\end{example}

In general this bound can be improved through a tighter control of the ratio $\text{Vol}(F)/\text{Vol}(F^\prime)$, in fact if this ratio is bounded below by a constant independent of dimension, then the relative error would be exponentially decreasing in $n$ and we would obtain a dimension-free result. 
However it is increasingly difficult to find a set $F^\prime$ to cover $F$ in high dimensions, thus this may not be feasible in practice. 

As noted by a reviewer, this example is part of the broader problem of estimating a small probability or proportion. Consider the problem of estimating the probability $\zeta_p := \mathbb{P}(X \in E_p) \in (0, 1]$ with the estimator $\bar{\zeta}_p = \sum_{i = 1}^n \mathbb{I}\{X_i \in E_p\}/n$ then:
\[ \mathbb{P}\left( \left| \frac{\bar{\zeta}_p - \zeta_p}{\zeta_p} \right| > \delta  \right) \leq 2\exp\left(-\frac{\delta^2 \zeta_p^2 n}{2\sigma^2} \right), \]
by Equation (\ref{eq:subG_Tail_rel}), where $X_i$ are iid copies of $X$. The requirement for relative consistency is satisfied if $\zeta_p^2 n \rightarrow \infty$, therefore depending on the speed at which $\zeta_p \rightarrow 0$ an arbitrarily fast or slow rate of growth of the function $n(p)$ is possible.


\section{Expectations of High-Dimensional Functions}
We consider two additional examples of Monte Carlo integration for expectations of smooth functions, going beyond the estimation of probability and proportions to demonstrate the general utility of the results presented in Section 2. 
The indicator function used in Section 3 is, in some sense, ill behaved, as it is not a continuous function and can change very rapidly in a small neighborhood. 
Considering functions with restrictions on their rate of growth through the Lipschitz property, along with other structural assumptions induces much faster rates of concentration.

\subsection{Lipschitz Convex or Concave Distributions on the Unit Square}

\begin{proposition}\label{prop:lip}
For convex or concave Lipschitz function $f_p(x): \mathbb{R}^p \rightarrow \mathbb{R}$ with Lipschitz constant $L_p$, and a random variable $X$ whose components are independently distributed and supported on $[0,1]$,
\begin{align*}
  \mathbb{P}\left\{ \left| \sum_{i = 1}^n f_p(X_i)/n - E[f_p(X)] \right| > \delta  \right\} 
  \leq 2\exp\left(-\frac{\delta^2 n}{2L_p^2} \right),    
\end{align*}
which implies the non-asymptotic $1- \alpha$ confidence interval 
\[ \left( \frac{\sum_{i = 1}^n f_p(X_i)}{n} - \sqrt{\frac{2\log(2/\alpha)L_p^2}{n}}, \ \frac{\sum_{i = 1}^n f_p(X_i)}{n} + \sqrt{\frac{2\log(2/\alpha)L_p^2}{n}}  \right)\] 
for $\alpha < 1$ and
\begin{align*}
  \mathbb{P}\left\{ \left| \frac{\sum_{i = 1}^n f_p(X_i)/n - E[f_p(X) ]}{E[f_p(X) ]} \right| > \delta  \right\} 
  \leq 2\exp\left\{-\frac{\delta^2 n E[f_p(X) ]^2}{2L_p^2} \right\},    
\end{align*}
where $X_i$ are iid copies of $X$.
\end{proposition}
The proof follows from the observation that a sum of Lipschitz convex/concave functions is still a Lipschitz convex/concave and the previously discussed concentration inequality. The rate of convergence for the relative accuracy only depends on the dimension through the Lipschitz constant $L_p$ and $E[f_p(X)]$. These can increase with dimension, which places constraints on the growth rate of $p$ relative to $n$, as demonstrated in Example 2.

\begin{example}
Consider 
\[
  \int_{[0,1]^p} f_p(x) dx = \int_{[0,1]^p} \frac{1}{1 + \norm{x}_q} dx,  
\] 
for any $q > 1$.
We apply Proposition \ref{prop:lip} to show that $n(p)$ has a polynomial dependence on $p$ for our relative consistency condition to hold. We first note that $f_p(x)$ is $p^{ \eta}$-Lipschitz, where $\eta = \max(1/q - 1/2, 0 )$ as  
\begin{align*}  
  |f_p(x) - f_p(y)| &=\left| \frac{1}{1 + \norm{x}_q} - \frac{1}{1 + \norm{y}_q} \right| \\
  &\leq \norm{x - y}_q \leq p^{ \eta} \norm{x - y}_2, 
\end{align*}
since $\norm{x}_{q} \leq \norm{x}_2 $ for all $ q \geq 2$, and $\norm{x}_q \leq p^{1/q - 1/2} \norm{x}_2$ for $ 1 \leq q < 2$.
The expectation of this function can be bounded from below by Jensen's inequality:
\[
  E\left( \frac{1}{1 + \norm{x}_q} \right) > \frac{1}{1 + E(\norm{x}_q)} \geq \frac{1}{1 + E(\norm{x}_1)} = \frac{2}{2 + p },  
\]
which results in:
\begin{align*}
  \mathbb{P}\left[ \left| \frac{\sum_{j = 1}^n f_p(X_i)/n - E[f_p(X) ]}{E[f_p(X) ]} \right| > \delta  \right] 
  \leq 2\exp\left(-\frac{\delta^2 n }{p^{2\eta}(2 + p)} \right),    
\end{align*}
meaning that it is sufficient for $n(p) = \omega (p^{2\eta + 1})$ to guarantee an asymptotically consistent estimator in relative accuracy.
\end{example}

\begin{remark}
  Proposition \ref{prop:lip} applies to any function defined on a hyper-rectangle as we may shift the domain of a function to the unit cube by centering and rescaling the function, although this will change the Lipschitz constant of the function.
\end{remark}

\subsection{Strongly Log-Concave Distributions}
For Lipschitz functions of strongly log-concave distributions the concentration can also be quite rapid, this allows us to consider an example of a function defined on a infinite domain and remove the convexity or concavity assumption on the function $f_p(x)$.

\begin{proposition}\label{prop:log}
  For a function $f_p(x): \mathbb{R}^p \rightarrow \mathbb{R}$ with Lipschitz constant $L_p$, and a strongly $\gamma$-log-concave random variable $X$,
  \begin{align*}
    \mathbb{P}\left\{ \left| \sum_{i = 1}^n f_p(X_i)/n - E[f_p(X) ] \right| > \delta  \right\} 
    \leq 2\exp\left(-\frac{\gamma\delta^2 n}{4L_p^2} \right),    
  \end{align*}
  which implies the non-asymptotic $1- \alpha$ confidence interval 
\[ \left( \frac{\sum_{i = 1}^n f_p(X_i)}{n} - \sqrt{\frac{4\log(2/\alpha)L_p^2}{\gamma n}},\ \frac{\sum_{i = 1}^n f_p(X_i)}{n} + \sqrt{\frac{4\log(2/\alpha)L_p^2}{\gamma n}}  \right)\] 
for $\alpha < 1$ and
  \begin{align*}
    \mathbb{P}\left\{ \left| \frac{\sum_{i = 1}^n f_p(X_i)/n - E[f_p(X) ]}{E[f_p(X) ]} \right| > \delta  \right\} 
    \leq 2\exp\left\{-\frac{\gamma\delta^2 n E[f_p(X) ]^2}{4L_p^2} \right\},    
  \end{align*}
  where $X_i$ are iid copies of $X$.
  \end{proposition}
This is shown by noting that an average of Lipschitz function is Lipschitz and the fact that the joint distribution of independent strongly $\gamma$-log-concave distribution is still strongly $\gamma$-log-concave, as shown in the Appendix.
One particularly important member of strongly log-concave distributions is the multivariate normal distribution with non-singular covariance matrix $\Sigma$, which is strongly $\gamma$-log-concave, with $\gamma$ equal to the minimal eigenvalue of $\Sigma^{-1}$. 

\begin{example}
  Consider the following expectation where $X \in \mathbb{R}^p$ follows a multivariate normal distribution with $0$ mean and a covariance matrix $\Sigma$ such that $\sigma_{\min}(\Sigma^{-1}) = \gamma > 0 $. This condition on the minimum eigenvalue implies that the distribution of $X$ is strongly $\gamma$-log-concave. We wish to estimate: 
  \[
\mathbb{E}\left[ \arctan\left( 1 + \norm{X}_1 \right) \right].\]
This function is $p^{1/2}$-Lipschitz due to the fact that: 
\begin{align*}
  \arctan(a) - \arctan(b) \leq \frac{a - b}{1 + ab},
\end{align*}
and the fact that $\lVert x \rVert_1 - \lVert y\rVert_1 \leq \lVert x - y \rVert_1$. 
Finally we may bound the expectation of this function from below by $\arctan(1)$, thus
\begin{align*}
  \mathbb{P}\left\{ \left| \frac{\sum_{i = 1}^n f_p(X_i)/n - E[f_p(X) ]}{E[f_p(X) ]} \right| > \delta  \right\} 
  \leq 2\exp\left[-\frac{\gamma\delta^2 n \arctan(1)^2}{4p} \right],    
\end{align*}
meaning that $n(p) = \omega(p)$ is needed to guarantee a relative consistent estimate. 
\end{example}

\section{Non-asymptotic and Asymptotic Confidence Intervals}

We briefly discuss and compare the non-asymptotic confidence intervals presented in Sections 3 and 4 to asymptotic and exact confidence intervals in a simple example to better understand their coverage properties.
The non-asymptotic confidence intervals considered in the previous sections will tend to be conservative for most distributions. 
This is necessary as the probability bounds used are applicable to a wide class of distributions and the equality in the bound cannot hold for all values of $t \in \mathbb{R}$, thus they must in general overstate the probability of the event. 

\begin{example}
  We observed a sample from a binomial distribution with parameters $(k, p)$ where $k$ is the number of coin tosses and $p$ is the probability of heads. We treat $k$ as known and consider the problem of providing a 90\% and 95\% confidence interval for $p$. We consider three possible confidence intervals, the Clopper-Pearson exact confidence interval, the Bayesian credible interval obtained with the $\text{Beta}(1/2,1/2)$ prior (which is Jeffrey's prior) and the non-asymptotic confidence interval from Proposition 1. The Bayesian credible interval is used instead of the usual normal approximation as for certain values of $p$, we are likely to obtain an observation of $0$ or $k$, which results in an invalid approximation.
  This credible interval also coincides with the standard normal approximation as $k \rightarrow \infty$ by the Bernstein von-Mises theorem \citep[Chapter 10]{van2000asymptotic}, and is therefore asymptotically valid.
  Empirical coverage probabilities from $1,000$ Monte Carlo simulations for a 90\% and 95\% nominal confidence intervals are available on Figures \ref{fig:90} and \ref{fig:95}. We see that the proposed non-asymptotic interval tends to be more conservative than the Clopper-Pearson interval. Although we do note the the Bayesian credible interval suffers poor coverage for certain values of the parameter $p$ even if the sample size is relatively high.
\end{example}

\begin{figure}
\centering
\begin{subfigure}{.5\textwidth}
  \centering
  \includegraphics[width=1\linewidth]{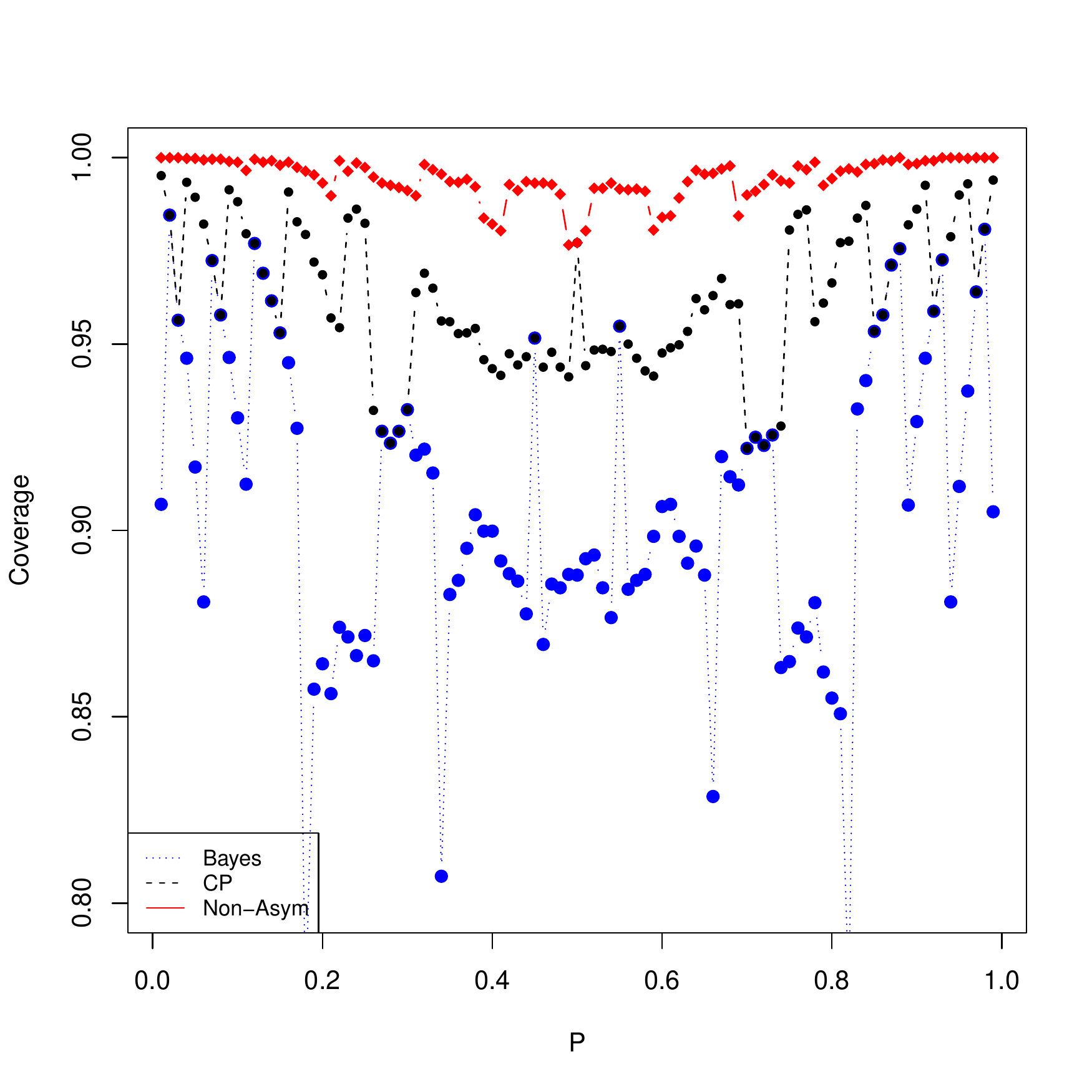}
  \caption{90\% coverage $k$ = 10}
  \label{fig:sub1}
\end{subfigure}%
\begin{subfigure}{.5\textwidth}
  \centering
  \includegraphics[width=1\linewidth]{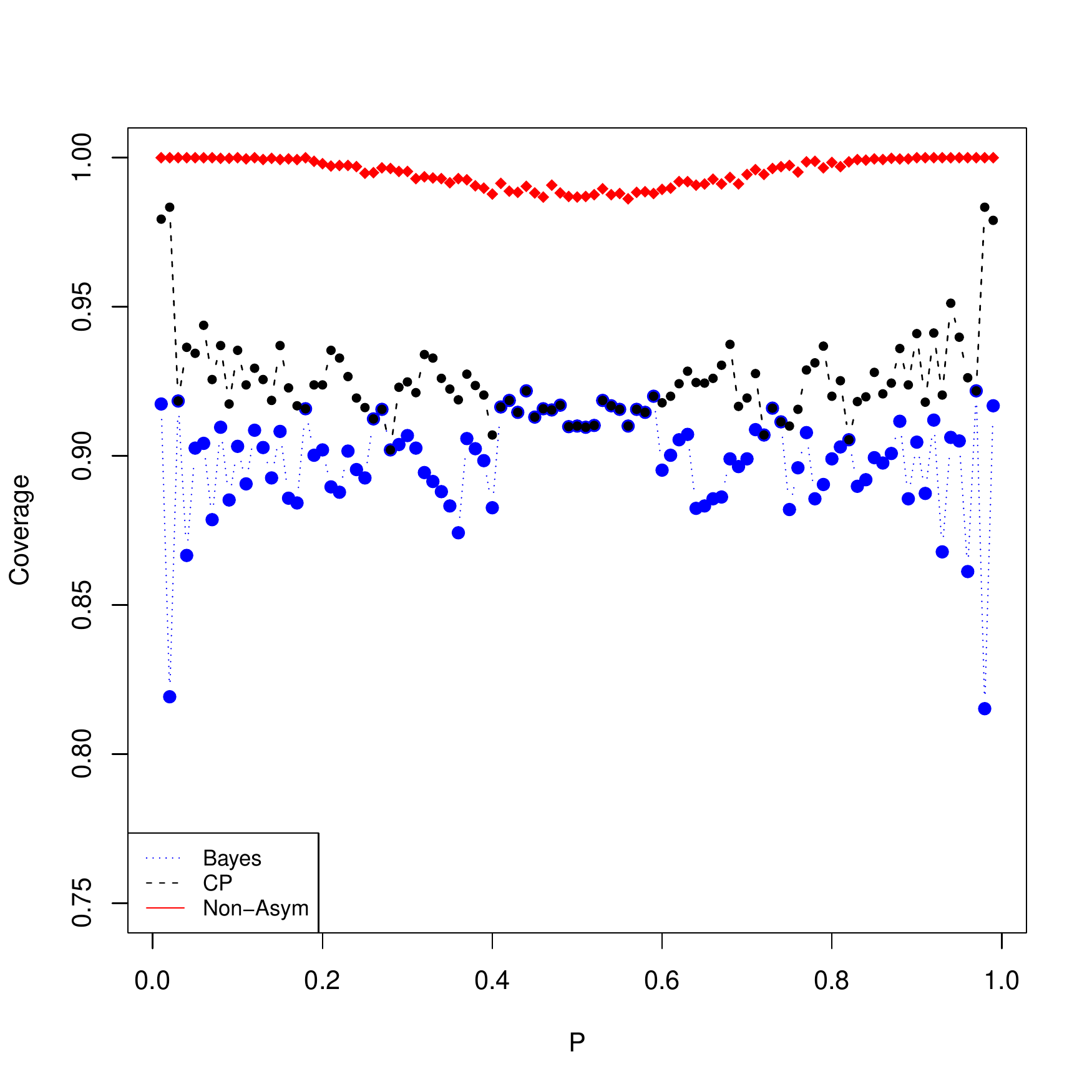}
  \caption{90\% coverage $k$ = 100}
  \label{fig:sub2}
\end{subfigure}
\caption[90\% confidence intervals]{Coverage probabilities for a nominal 90\% confidence interval obtained from $1,000$ Monte Carlo replications, for $k = 10$ on the left and $k = 100$ on the right. Non-asymptotic coverage probabilities are given in red, Clopper-Pearson exact confidence interval coverage given in black and the Bayesian credible interval is given in blue.}
\label{fig:90}
\end{figure}

\begin{figure}
\centering
\begin{subfigure}{.5\textwidth}
  \centering
  \includegraphics[width=1\linewidth]{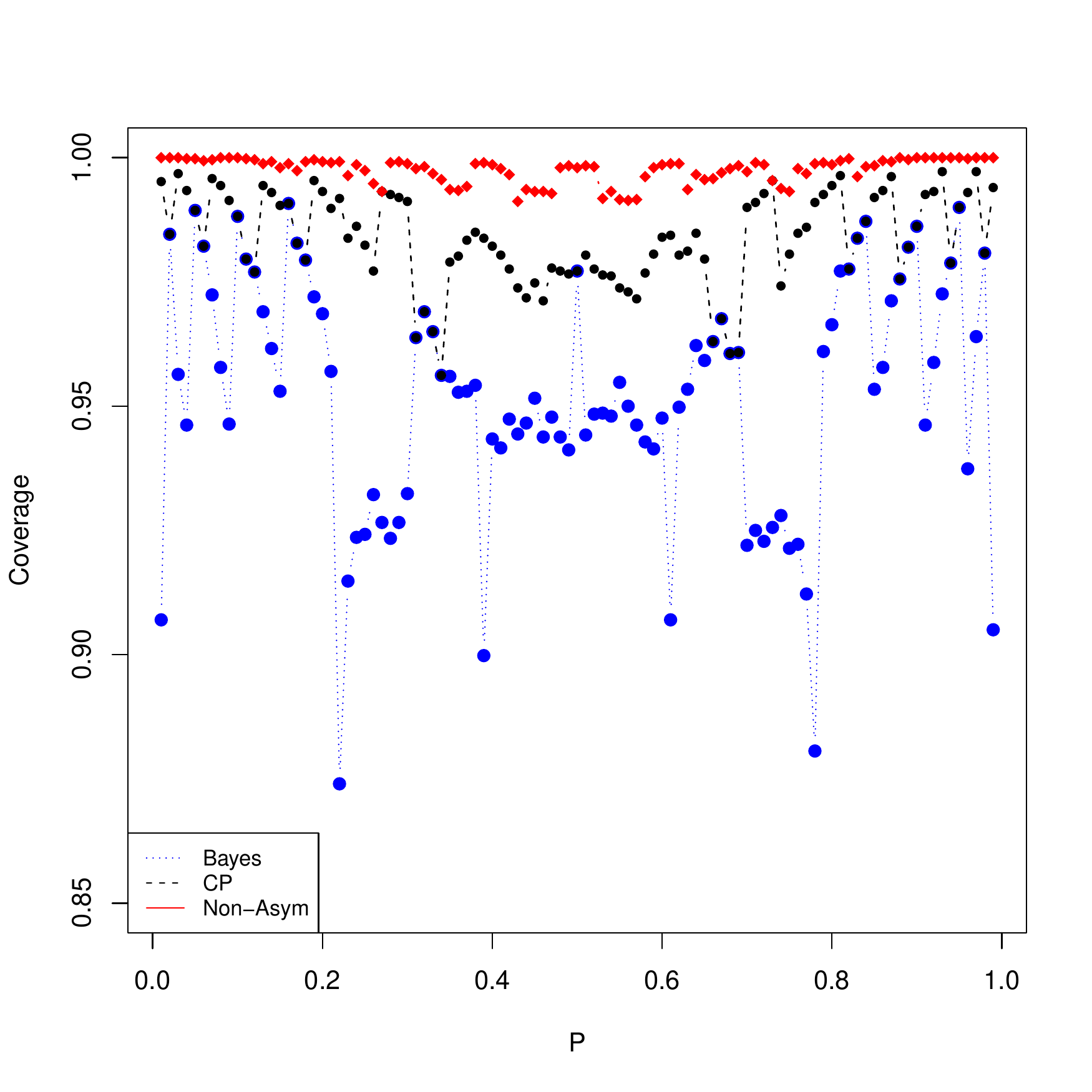}
  \caption{A subfigure}
\end{subfigure}%
\begin{subfigure}{.5\textwidth}
  \centering
  \includegraphics[width=1\linewidth]{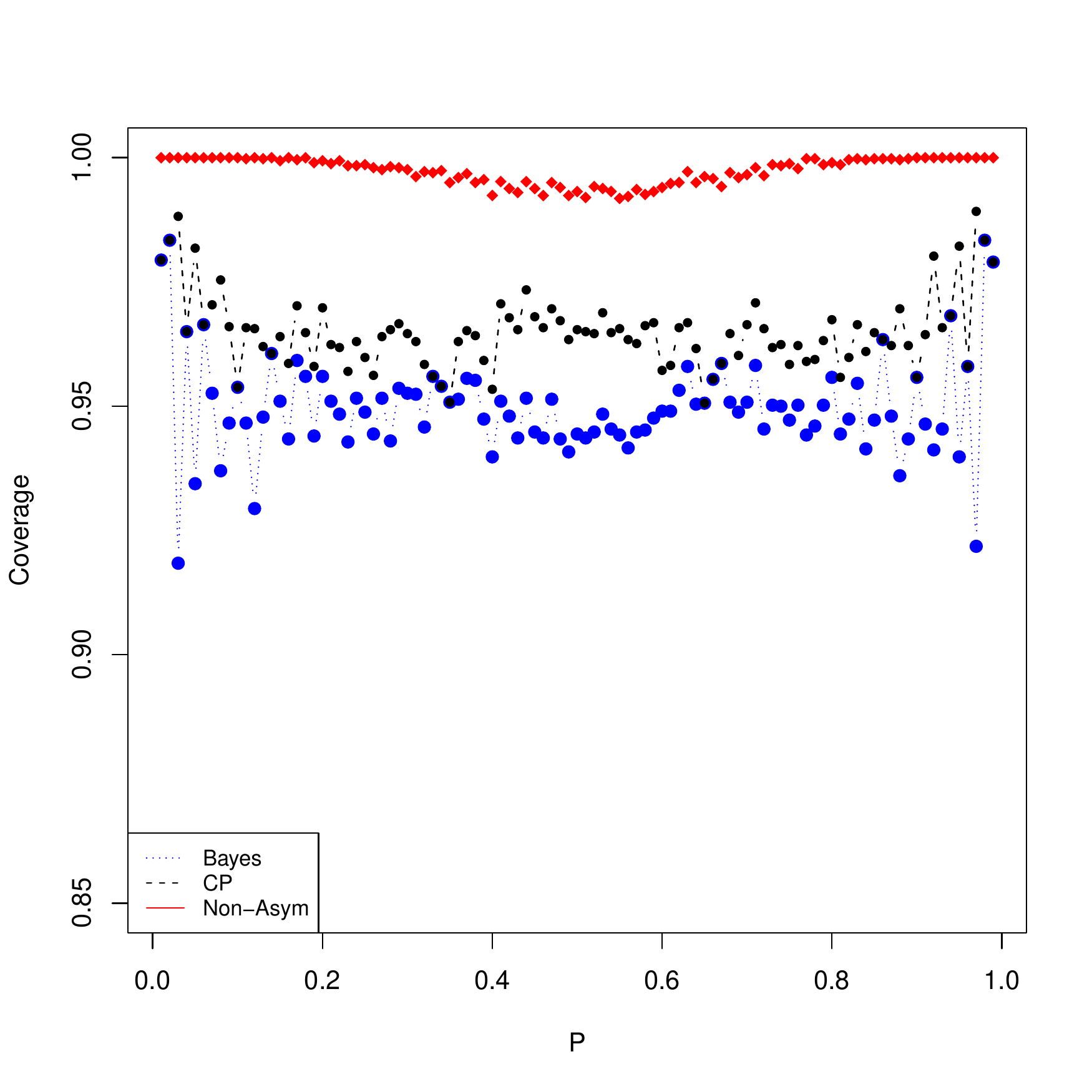}
  \caption{A subfigure}
\end{subfigure}
\caption[95\% confidence intervals]{Coverage probabilities for a nominal 95\% confidence interval obtained from $1,000$ Monte Carlo replications, for $k = 10$ on the left and $k = 100$ on the right. Non-asymptotic coverage probabilities are given in red, Clopper-Pearson exact confidence interval coverage given in black and the Bayesian credible interval is given in blue.}
\label{fig:95}
\end{figure}

In general, if sufficient knowledge of the distribution is available to the practitioner, it is possible to obtain better alternatives to the non-asymptotic intervals by performing a bespoke analysis. 
In the binomial example, exact intervals are available as we have the analytical form of the probability mass function.
However it is not always the case that such knowledge available to us.
Although, we do note that if an exact interval is unavailable, the asymptotic confidence intervals may produce invalid anti-conservative statements while the non-asymptotic intervals provides valid conservative statements.

\section{Discussion}
We have shown that the applicability of Monte Carlo integration varies substantially depending on the specific combination of the function being integrated and the measure it is integrated against, and we relate the approximation error to the rate of concentration of a sum of iid random variables.
We note that the procedure can be catastrophically bad for approximating volumes, but for integrating specific functions it may perform well.
The analysis can potentially be extended to sub-exponential random variables, or random variables defined on manifolds by exploiting the numerous available concentration inequalities.
Another important extension would be to consider non iid random variables, this could help in providing non-asymptotic bounds on the performance of Monte Carlo Maximum Likelihood Estimates for exponential graphical models as introduced in \cite{geyer1992constrained} and general Monte Carlo maximum likelihood estimation approaches.
This can potentially be done through the bounds of strongly log-concave densities presented in Section 4.2, however we would need to ensure that the sampled distribution is strongly log-convex, which is a non-trivial assumption to verify.

There are also other variants of Monte Carlo which are used in high dimensions.
For example, it has been observed that quasi-Monte Carlo performs well in some financial applications in high dimensions \citep{paskov1996faster}, even if there is an explicit dependence on the dimensionality of the function on the absolute error scale: $O[\log(n)^p/n]$.
However this rate is obtained through a worst case scenario analysis and the true dependence on dimension may be much smaller. Determining the effective dimension of the problem is still an open question, see \cite{owen2021open}.


\bigskip
\begin{center}
{\large\bf Appendix}
\end{center}
\section*{Properties of Lipschitz and Log-Concave Distributions}
\begin{lemma}\label{prop:jensen}
  \begin{itemize}
    \item If $f_i(x_i): \mathbb{R}^p \rightarrow \mathbb{R}$ are $L_p$-Lipschitz functions with respect to the Euclidean norm for $i = 1, \dots, n$, then $g(x_1, \dots, x_n) =\sum_{i = 1}^n f_i(x_i)/n$ is a $L_p/n^{1/2}$-Lipschitz function. 
    \item If $f_i(x_i): \mathbb{R}^p \rightarrow \mathbb{R}$ are $L_p$-Lipschitz convex/concave functions with respect to the Euclidean norm for $i = 1, \dots, n$, then $g(x_1, \dots, x_n) =\sum_{i = 1}^n f_i(x_i)/n$ is a $L_p/n^{1/2}$-Lipschitz convex/concave function. 
  \end{itemize}
\end{lemma}
\begin{proof}
  Let $x_i, y_i \in \mathbb{R}^p$, then
  \begin{align*}
    \left|\frac{1}{n}\sum_{i = 1}^n f_i(x_i) - \frac{1}{n}\sum_{i = 1}^n f_i(y_i) \right| &\leq \frac{1}{n} \sum_{i = 1}^n |f_i(x_i) - f_i(y_i)|\\ 
    &\leq L_p \frac{\sum_{i = 1}^n \norm{x_i - y_i}_2}{n} \\
    &\leq L_p \frac{\sum_{i = 1}^n [\norm{x_i - y_i}^2_2 ]^{1/2}}{n}  \\
    &\leq \frac{L_p}{n^{1/2}} \norm{(x_1, \dots, x_n) - (y_1, \dots, y_n)}_2, 
  \end{align*} 
  where the last inequality follows from Jensen's inequality for the concave function $\rho(x) = x^{1/2}$.
The second statement follows from the first by noting that sums of convex/concave functions remain convex/concave.
\end{proof}

This basic property of strongly log-concave distribution is shown in \cite{saumard2014log}, but due to notational differences we provide a proof for the convenience of the reader. 

\begin{lemma}\label{prop:log-concave}
Let $X_i$ be independent copies of a $\gamma$ strongly log-concave distribution for $i = 1, \dots, n$, then the joint distribution $(X_1, \dots, X_n)$ is also a $\gamma$ strongly log-concave distribution.
\end{lemma}
\begin{proof}
  Let $X_i$ be distributed according to $\phi_i(x_i) = \exp[-\psi_i(x_i)]$
  Let $y_i, z_i \in \mathbb{R}^p$, for $i = 1, \dots, n$ then by the log-concavity of each $x_i$, for each $i = 1, \dots, n$ and all $\lambda \in [0,1]$
  \begin{align*}
    \lambda\psi_i(y_i) + (1- \lambda)\psi_i(z_i) - \psi_i\{ \lambda y_i + (1 - \lambda)z_i  \} \geq \frac{\gamma}{2} \lambda(1- \lambda)\norm{y_i - z_i}_2^2, 
    \\
  \end{align*} 
  which implies
  \begin{align*}
    &\lambda \sum_{i =1}^n \psi_i(y_i) + (1- \lambda)\sum_{i =1}^n\psi_i(z_i) - \sum_{i =1}^n\psi_i[ \lambda y_i + (1 - \lambda)z_i  ] \\&\geq \frac{\gamma}{2} \lambda(1- \lambda) \sum_{i =1}^n\norm{y_i - z_i}_2^2\\
    &\geq \frac{\gamma}{2} \lambda(1- \lambda) \norm{(y_1, \dots, y_n ) - (z_1, \dots, z_n) }_2^2, 
  \end{align*}
  where the final inequality follows from the triangle inequality.
\end{proof}

\section*{Proof of Propositions \ref{prop:lip} and 3}
\subsection*{Proof of Proposition \ref{prop:lip}}

The first and third part of the proposition follow from choosing the function $g(X_1, \dots, X_n) = \sum_{i = 1}^n f(X_i)/n$ in bound (5), and noting by Lemma \ref{prop:jensen} that $g$ is convex/concave with Lipschitz constant $L_p/n^{1/2}$.
The confidence interval is obtained by setting: 
\[
 \alpha = 2\exp\left( -\frac{\delta^2n}{2L_p^2} \right) ,
\]
and solving for $\delta$. 

\subsection*{Proof of Proposition 3}
The proof for the first and third statements is almost the same as that of Proposition 2, we let $g(X_1,\dots, X_n) = \sum_{i = 1}^n f(X_i)/n$ and note that this function is $L_p/n^{1/2}$-Lipschitz and combine this with the fact that the joint distribution of $(X_1, \dots, X_n)$ is also strongly $\gamma$-log-concave (as shown in Lemma 2) gives the desired result. The confidence interval is obtained by setting: 
\[
 \alpha = 2\exp\left( -\frac{\gamma\delta^2n}{4L_p^2} \right) ,
\]
and solving for $\delta$.
\bibliography{biblio}
\end{document}